\documentclass{llncs}
\usepackage{amssymb,amsmath}
\usepackage{graphicx}
\usepackage{color}
\usepackage{epstopdf}
\usepackage{subfigure}
\newcommand{\hide}[1]{}
\newcommand{\QED}{\hfill$\qed$}
\newtheorem{procedure}{Procedure}
\newtheorem{algorithm}{Algorithm}

\usepackage{lineno}

\begin{document}

\title{Geometric $k$-Center Problems with
Centers Constrained to Two Lines}
\author{Binay Bhattacharya\inst{1}
\and Ante \'Custi\'c\inst{2}
\and Sandip Das\inst{3}
\and Yuya Higashikawa\inst{4} 
\and Tsunehiko Kameda\inst{1}
\and Naoki Katoh\inst{5}
}

\institute{School of Computing Science, Simon Fraser University, Burnaby, Canada 
\and Dept. of Mathematics, Simon Fraser University, Burnaby, Canada
\and Advanced Computing and Microelectronics Unit, Indian Stat. Inst., Kolkata, India 
\and Dept. of Information and System Engineering, Chuo University, Tokyo, Japan
\and School of Science and Technology, Kwansei University, Hyogo, Japan
}

\maketitle
\begin{abstract}
We consider the $k$-center problem,
in which the centers are constrained to lie on two axis-parallel lines.
Given a set of $n$ weighted points in the plane,
which are sorted according to their $x$-coordinates,
we show how to test in $O(n\log n)$ time if $k$ piercing points placed on
two lines, parallel to the $x$-axis,
can pierce all the disks of different radii centered at the $n$ given points.
This leads to an $O(n\log^2 n)$ time algorithm for the weighted $k$-center problem.
We then consider the unweighted case,
where the centers are constrained to be on two perpendicular lines.
Our algorithms run in $O(n\log^2 n)$ time in this case also.
\end{abstract}

\section{Introduction}\label{sec:intro}
The {\em $k$-center problem}, a.k.a., the {\em $p$-center problem},
is one of the most intensively studied problems in computational geometry.
The {\em geometric} $k$-center problem is defined as follows~\cite{preparata1990}.
Given are a set $P =\{p_1,p_2,\ldots, p_n\}$ of $n$ points in the plane,
where point $p_i$ ($i=1,2,\ldots, n$) has weight $w_i$,
and a positive integer $k$.
The objective is to find a set of $k$ centers such that the maximum
over all points in $P$ of the weighted distance from a point to its nearest center
is minimized. 
It was shown by Megiddo \cite{megiddo1984} that this problem is NP-hard in its general form.
To find more tractable cases solvable in polynomial time,
researchers have considered many variations of this problem in terms of the metric used,
the number of centers,
the weights (uniform vs. non-uniform), and constraints on the allowed positions
(discrete points, lines, polygon, etc.) of the centers~\cite{hurtado2000,megiddo1983b,megiddo1983d}.

When the centers are constrained to a single line,
Karmakar et al. showed that the {\em unweighted} $k$-center problem,
where the vertices have the same weight,
can be solved in $O(n\log n)$ time~\cite{karmakar2013},
and Wang and Zhang showed that the {\em weighted} $k$-center problem,
can also be solved in $O(n\log n)$ time~\cite{wang2014a}.
It was recently shown by Chen and Wang \cite{chen2015a} that the weighted version of this problem
has a lower bound of $\Omega(n\log n)$ on its time complexity.
The $O(n\log n)$ time algorithms mentioned above are based on Megiddo's parametric
search~\cite{megiddo1983a} with Cole's speed-up~\cite{cole1987}.
Since they use the AKS network~\cite{ajtai1983}, 
the coefficient hidden in the big-O notation is huge.
If a more practical sorting network, such as {\em bitonic sorting network}~\cite{knuth1997},
is used the time requirement increases to $O(n\log^2 n)$.
The 1-dimensional $k$-center problem is discussed by Bhattacharya et al.~\cite{bhattacharya2007}, Chen and Wang \cite{chen2015a}, and Fournier and Vigneron~\cite{fournier2013}.
If the $l_1$  or $l_{\infty}$ distance metric is used instead of $l_2$,
in other words, if the enclosing shape is an axis-parallel square,
then we can apply a method due to Frederickson and Johnson
\cite{frederickson1982,frederickson1983}
to solve this problem in $O(n)$ time~\cite{fournier2013},
provided the points are presorted.

In this paper we consider the $k$-center problem constrained to two parallel or perpendicular lines.
We first show that the weighted $k$-{\em piercing} problem on two parallel lines
can be solved in $O(n\log n)$ time.
This implies that $k$-center problem can be solved in $O(n\log^2 n)$ time,
using Megiddo's parametric search~\cite{megiddo1983a} with Cole's speed-up~\cite{cole1987}.
In the unweighted case and under the $l_1$  or $l_{\infty}$ distance metric,
it is known that this problem can be solved in $O(n\log n)$ time without using a huge
sorting network~\cite{bereg2015b}.

In the line-constrained problem,
it is a standard practice to use the intersection of an object $o(p_i)$ around each point
$p_i$ and a line.
This object $o(p_i)$ is a square under the $l_{\infty}$ distance metric and a disk under
the $l_2$ metric,
which is the case in our model. 
When disks are used instead of squares, 
some results from \cite{bereg2015b} carry over,
but a few complications arise.
First,
the end points of the intersection intervals do not have the same order as
 their $x$-coordinates.
Second,
if a disk intersects both of the two lines,
the intersection points do not have the same $x$-coordinate in general.
These complications make the problem more challenging.

The rest of the paper is organized as follows.
Sec.~\ref{sec:2Lines} discusses the case where the centers are constrained to two parallel lines,
and presents an $O(n\log^2 n)$ time algorithm.
In Sec.~\ref{sec:perpendicular} we discuss the unweighted case,
where the centers are constrained to two perpendicular lines.
We then show that the problem can be solved in $O(n\log^2 n)$ time.
Finally, Sec.~\ref{sec:conclusion} concludes the paper with a summary and open problems.

\section{Centers constrained to two parallel lines}\label{sec:2Lines}

\subsection{Preliminaries}\label{subsec:prelim}

Let $P=\{p_1, p_2,\ldots, p_n\}$ be a set of points in the plane,
where point $p_i$ has a positive weight $w_i$.
We denote the $x$-coordinate (resp. $y$-coordinate) of point $p_i$ by  $p_i.x$ (resp. $p_i.y$),
and assume that $p_i.x \leq p_{i+1}.x$ holds for $i=1,2,\ldots, n-1$,
i.e., the points are sorted according to their $x$-coordinates.
If the $x$-coordinates of two points are the same,
they can be ordered arbitrarily.
Let $d_i(r)$ denote the disk with radius $r/w_i$ centered at point $p_i$.
Let $L_1$ and $L_2$ be the two given axis-parallel horizontal lines ($L_1$ above $L_2$).
A problem instance is said to be {\em $(k,r)$-feasible} if $k$ centers can be placed
on the two lines in such a way that there is a center within distance $r/w_i$
from every point $p_i$.
For $i= 1, 2,\ldots, n$ let $J_1^i(r)= d_i(r)\cap L_1$ (resp. $J_2^i(r)= d_i(r)\cap L_2$).
We assume that at least one of $J_1^i(r)$ and $J_2^i(r)$ is non-empty,
since otherwise the problem instance is not $(k,r)$-feasible.
If $J_1^i(r)\not=\emptyset$ and $J_2^i(r)\not=\emptyset$, then
they are called {\em buddy intervals} or just {\em buddies}, and the corresponding point $p_i$ is called a \emph{buddy point}.
So, given a set $P$ of weighted points and a radius $r$, we are interested in finding a minimum cardinality set $S$ of points on $L_1$ and $L_2$ such that $S\cap J_1^i(r)\neq\emptyset$ or $S\cap J_2^i(r)\neq\emptyset$, for every $i=1,2,\ldots,n$. We call such $S$, a set of piercing points. From now on, up to Lemma~\ref{lem:feasibility1}, we assume that radius $r$ is fixed. 

Given $i$, $1\leq i\leq n$, we call a set $S_i$ of piercing points for the disks around points $P_i=\{p_1,p_2,\ldots,p_i\}$ a \emph{partial solution} on $P_i$. If $P_i=P$, a partial solution is a \emph{complete solution}.
Given a partial solution $S_i$ with $|S_i|=z$ piercing points, let $I_1$ be an interval on $L_1$ that represents the section along which the rightmost point of $S_i\cap L_1$ can be moved so that all the disks for the points in $P_i$ are still pierced. Analogously, let $I_2$ represents such an interval on $L_2$. Then the triple $c=(I_1,I_2;z)$ is said to be a \emph{configuration} for $P_i$, where the non-negative integer $z$ is called the {\em count} of the configuration~\cite{bereg2015b}. Furthermore, associated with configuration $c$ is a {\em piercing sequence} $\wp(c)$ of $\max\{z-2,0\}$ fixed piercing points
on $L_1$ and $L_2$, i.e., $S_i$ without rightmost points on $L_1$ and $L_2$. In other words, configuration $(I_1,I_2;z)$ represents a class of partial solutions that consist of $z$ piercing points that are identical,  except for the rightmost points on $L_1$ and $L_2$, which can be anywhere on $I_1$ and $I_2$, respectively. 

For the purpose of our algorithm, we would like to be able to dismiss partial solutions/configurations that cannot be the unique ones that lead to the complete solutions of minimum cardinality. To that end, we introduce the domination property. Given a partial solution $S$, a \emph{complete extension} of $S$ is a superset of $S$ which is a complete solution.
Let $c'=(I'_1,I'_2;z')$ and $c''=(I''_1,I''_2;z'')$ be two different configurations for  $P_i$. We say that $c'$ \emph{dominates} $c''$ if, regardless of what the remaining points in $P$ are,  there cannot exist a complete extension of some partial solution represented by $c''$ which has a strictly smaller cardinality than all the complete extensions of all partial solutions represented by $c'$.

\subsection{Algorithm}\label{sec:algorithm}

Now we are ready to present the algorithm. We scan the points in $P=\{p_1, p_2,\ldots,\linebreak p_n\}$ from left to right, such that in step $i$ we generate all non-dominated configurations for $P_i$. We maintain the set $F$ of such configurations, called the {\em frontier configurations}. After we scan all the $n$ points, we know that the given instance is $(k,r)$-feasible if and only if there is a configuration in $F$ with the count at most $k$.
Given an interval $I$ on either $L_1$ or $L_2$, let $l(I)$ denote the left endpoint and $e(I)$ the right endpoint of $I$.

\begin{algorithm}\label{alg:feasibility}
~
\begin{enumerate}
\item
Initialize $F=\big\{([-\ell,-\ell],[-\ell,-\ell]; 0)\big\}$, where $\ell$ is a very large number.
\item
For $i=1,2,\ldots, n$, execute Steps 3, 4 and 8.
\item
Set $F'=\emptyset$.
\item
For each configuration $c=(I_1,I_2; z) \in F$,
do Steps 5--7 that apply,
and put the generated configurations in $F'$.
\item
~$[$$J_1^i(r)=\emptyset \wedge J_2^i(r)=\emptyset$$]$
The problem instance is not $(k,r)$-feasible (no point can pierce $d_i(r)$). Stop
\item
~$[$$J_1^i(r)\not=\emptyset$$]$
If $I_1 \cap J_1^i(r) \not= \emptyset$ then
convert $c$ into $(I_1\cap J_1^i(r),I_2; z)$
else convert into $(J_1^i(r),I_2; z+1)$ and add $e(I_1)$ into $\wp(c)$.
\item
~$[$$J_2^i(r)\not=\emptyset$$]$
If $ I_2 \cap J_2^i(r)\not= \emptyset$ then
convert $c$ into $(I_1,I_2\cap J_2^i(r); z)$
else convert into $(I_1,J_2^i(r); z+1)$ and add $e(I_2)$ into $\wp(c)$.
\item
Remove the dominated configurations from $F'$, and replace $F$ by $F'$.
\item
The problem instance is $(k,r)$-feasible if and only if there is
a configuration in $F$ whose count is no more than $k$.
\QED
\end{enumerate}
\end{algorithm}

\begin{theorem}\label{thm:correctness}
Algorithm~\ref{alg:feasibility} is correct.
\end{theorem}
\begin{proof}[Sketch]
In Step 3 of the algorithm, for all $i$ we aim to calculate a set $F$ of partial solutions on points $P_i$, such that, no mater what the rest of $P$ is, there is one partial solution in $F$ that can be extended to a complete solution with the minimum cardinality. If that is the case, the algorithm is correct. 

Recall that we group partial solutions into configurations and treat them as a unit. 
One such configuration $c=(I_1, I_2; z)$ represents all partial solutions consisting of the points
in $\wp(c)$ and one piercing point each on $I_1$ and $I_2$.
Consider configuration $c=(I_1,I_2;z)$, and a new point $p_{i}$ with the corresponding interval $J_1^i(r)$. 
If $J_1^i(r)$ is already pierced by a point from $\wp(c)$, then from the definition of $\wp(c)$, we have $l(J_1^i(r))\leq l(I_1)$. In this case, from the fact that $p_{i}.x\geq p_j.x$ for $i> j$, we also have $r(J_1^i(r))\geq r(I_1)$. It thus follows that $I_1\subseteq J_1^i(r)$, which implies that $J_1^i(r)$ is pierced also by every point from $I_1$.
Hence we know whether $J_1^i(r)$ can be pierced by some partial solution corresponding to $c$ without observing $\wp(c)$. Therefore, we don't lose any crucial information by restricting to configurations, i.e., only the number of piercing points (count) and positions of the rightmost piercing points on $L_1$ and $L_2$ play a role in the quality of a partial solution with respect to the remaining points in $P$. 

In Step 6 of the algorithm, if $I_1 \cap J_1^i(r) \not= \emptyset$, only $(I_1\cap J_1^i(r),I_2; z)$, and not  $(J_1^i(r)\setminus I_1,I_2; z+1)$, is generated. That is because the latter is dominated by the former. When creating $(J_1^i(r),I_2; z+1)$, in addition we need to add an arbitrary point from $I_1$ to $\wp(c)$. Step 7 works analogously.

From the definition of domination, it follows that removing dominated configurations in Step 8 won't affect the correctness.
\QED
\end{proof}

By implementing Algorithm~\ref{alg:feasibility} directly, inefficiencies would be caused if
searching $F$ in Step~4 and removing dominated configurations in Step 8 are blindly executed. In the following subsection we describe how Algorithm~\ref{alg:feasibility} can be implemented efficiently.

\subsection{Implementation}\label{sec:implement}

Let us first discuss how to implement one round of Steps 4 and 8 of Algorithm~\ref{alg:feasibility} efficiently.
For this purpose we identify the intervals $I_1$, $I_2$ of configuration $(I_1,I_2;z)$ with its right endpoints $e(I_1)$, $e(I_2)$, respectively. We can do so since $I_1$ intersects $J_1^{i}(r)$ if and only if $l(J_1^{i}(r))\leq e(I_1)$, in which case $e(I_1 \cap J_1^{i}(r))=\min\{e(I_1),e(J_1^{i}(r))\}$. This follows from  $l(I_1)<e(J_1^{i}(r))$, which is a consequence of $p_{i}.x\geq p_j.x$ for $i> j$ and the fact that $l(I_1)$ is the left endpoint of some interval that has been processed so far. So in the rest of this section, a configuration will be represented by $(x_1,x_2;z)$, where $x_1$ and $x_2$ are the right endpoints of the respective intervals $I_1$ and $I_2$.

To represent configuration $c=(x_1, x_2; z)$ visually,
we draw a line segment between $x_1$ and $x_2$ and label it by its count $z$.
We say that configurations $(x'_1, x'_2; z')$ and $(x''_1, x''_2; z'')$
 {\em cross each other} if 
either $x'_1 < x''_1$ and $x'_2 > x''_2$ or $x'_1 > x''_1$ and $x'_2 < x''_2$ hold.
Configurations $(x'_1, x'_2; z')$ and $(x''_1, x''_2; z'')$ are said to be {\em disjoint}
if either $x'_1 < x''_1$ and $x'_2 < x''_2$ or $x'_1 > x''_1$ and $x'_2 > x''_2$ hold. See Fig.~\ref{fig:dominated2}.
\begin{figure}[h]
\vspace{-5pt}
\centering
\subfigure[]{\includegraphics[height=20mm]{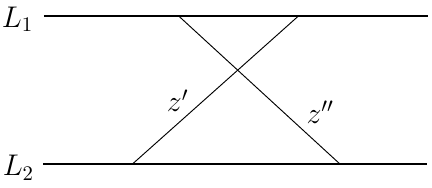}}
\hspace{1cm}
\subfigure[]{\includegraphics[height=20mm]{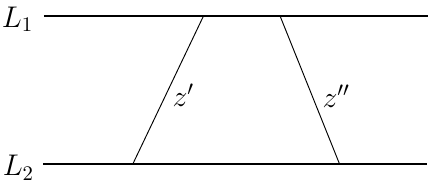}}
\vspace{-5pt}
\caption{Representations of configurations:
(a) configurations cross each other
(b) configurations are disjoint.}
\label{fig:dominated2}\vspace{-35pt}
\end{figure}
\begin{figure}[h]
\centering
\includegraphics[height=20mm]{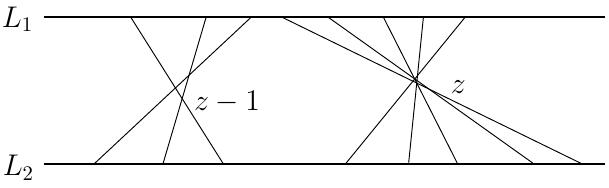}
\vspace{-7pt}
\caption{Set of configurations without any dominating configuration}
\label{fig:nondomF}
\vspace{-17pt}
\end{figure}

\begin{lemma}\label{lem:confSet}
Let $F$ be a set of configurations that represent partial solutions for point set $P_i$
 where no configuration in $F$ dominates another. 
Then
\begin{enumerate}
\item[(a)]
The counts of the configurations in $F$ differ by at most one.
\item[(b)]
Every pair of configurations in $F$ cross each other if they have the same count.
\item[(c)]
Every pair of configurations $(x'_1, x'_2; z),\ (x''_1, x''_2; z-1)\in F$  are disjoint and $x'_1>x''_1$, $x'_2>x''_2$. (See Fig.~\ref{fig:nondomF}.)
\end{enumerate}
\end{lemma}

\begin{proof}
(a) Let $c'=(x'_1, x'_2; z')$ and $c''=(x''_1, x''_2; z'')$ be two configurations for the points in $P_i$,
such that $z'\geq z''+2$. Let $S$ be a minimum cardinality set of piercing points that extends $\wp(c')$ to a complete solution. Then $S\cup\{x''_1,x''_2\}$ extends $\wp(c'')$ to a complete solution also. Since the cardinality of the latter complete solution doesn't exceed the cardinality of the former one, we have that $c''$ dominates $c'$.

(b) It suffices to prove that if $c'=(x'_1, x'_2; z)$ and $c''=(x''_1, x''_2; z)$ don't dominate each other, then they must cross each other. To prove the contrapositive, assume that they don't cross each other. Then without loss of generality we can assume that $x'_1\leq x''_1$ and $x'_2\leq x''_2$. 
But then from arguments in the proof of Theorem~\ref{thm:correctness},
it easily follows that $c'$ cannot lead to a complete solution of a smaller cardinality, i.e., $c''$ dominates $c'$.

(c) Assume to the contrary that there are two disjoint configurations $c'=(x'_1, x'_2; z)$ and $c''=(x''_1, x''_2; z-1)$ in $F$
such that, without loss of generality, $x'_2\leq x''_2$ holds.
Then again, if $S$ is a minimum cardinality extension of $\wp(c')$ to a complete solution, then $S\cup \{x''_1\}$
is an extension of $\wp(c'')$ to a complete solution, hence $c''$ dominates $c'$,
a contradiction to the fact that no configuration in $F$ dominates another.
\QED
\end{proof}

Lemma~\ref{lem:confSet} gives us a description of frontier configurations $F$. Now we describe how Steps 4 and 8 of Algorithm~\ref{alg:feasibility} can be implemented efficiently, i.e., how $F$ can be recomputed when a new point $p_{i}$ and corresponding intervals $J^i_1(r)$ and $J^i_2(r)$ are introduced.

Let the new point $p_i$ be a non-buddy point such that $[j^i_l,j^i_r]=J^i_1(r)\neq \emptyset$ and $J^i_2(r)=\emptyset$. Now we observe which configurations from $F$ will be preserved, which will be modified and which discarded, after $J^i_1(r)$ is introduced. We partition $F$ into three groups. Let $F_L=\{(x_1,x_2;z)\in F \mid x_1 <j^i_l\}$, $F_M=\{(x_1,x_2;z)\in F \mid j^i_l\leq x_1 \leq j^i_r\}$ and $F_R=\{(x_1,x_2;z)\in F \mid j^i_r<x_1 \}$. 
Step 6 of Algorithm~\ref{alg:feasibility} replaces each configuration $(x_1,x_2;z)$ from $F_L$ with $(j^i_r,x_2;z+1)$. Let $\bar{c}=(\bar{x}_1,\bar{x}_2,z)$ be an element of $F_L$ with the smallest first component. By a similar argument to that in the proof of Lemma~\ref{lem:confSet}, one can easily see that $(j^i_r,\bar{x}_2;z+1)$, created from $\bar{c}$, will dominate all other configurations created from $F_L$.
Analogously, configuration $(j^i_r,\hat{x}_2;z)$, created in Step~6 from a configuration $(\hat{x}_1,\hat{x}_2;z)$ from $F_R$ with the smallest first component, will dominate all other configurations created from the elements of $F_R$. Lastly, Step~6 leaves all elements of $F_M$ unchanged. To summarize, in the case of non-buddy point $p_i$ with $J^i_1(r)\neq \emptyset$, we can update frontier configurations in $F$ by scanning it from left and right with respect to the first component, and removing all components until we hit the interval $J^i_1(r)$. 
Furthermore, in place of the removed configurations, two new configurations are considered to join $F$.
This is illustrated in Fig.~\ref{fig:nonBuddyBuddy}(a), where the solid lines represent the frontier configurations
that are in $F$ before $p_i$ is processed. The configurations that survive are bolded. Two dashed lines are new configurations  considered, where in this case, the right one dominates, and hence is admitted to $F$.

\begin{figure}[h]
\vspace{-15pt}
\centering
\subfigure[]{\includegraphics[height=20mm]{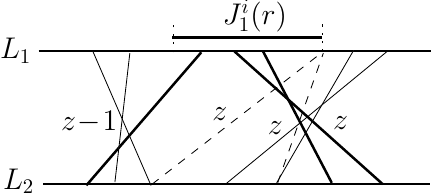}}
\hspace{8mm}
\subfigure[]{\includegraphics[height=26mm]{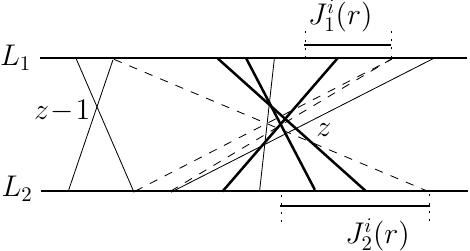}}
\vspace{-5pt}
\caption{Updating the frontier configurations: (a) for a non-buddy interval on $L_1$;
(b) for buddy intervals.
}
\label{fig:nonBuddyBuddy}
\vspace{-10pt}
\end{figure}
We present the above implementation more formally as a procedure.

\begin{procedure}\label{proc:nonBuddyL1} {\tt NonBuddy$(J_1^i(r))$}
\begin{enumerate}
\item
Scan the configurations in $F$ from left to right and right to left on $L_1$, until  $l(J_1^i(r))$ and $e(J_1^i(r))$ are reached, respectively. Delete all scanned configurations (except those with the first component $l(J_1^i(r))$ or $e(J_1^i(r))$).
\item
Let $(\bar{x}_1,\bar{x}_2;\bar{z})$ and $(\hat{x}_1,\hat{x}_2;\hat{z})$ be the deleted configurations with smallest first components in the left to right and right to left scanning process, respectively. (Note that they do not exist if $F_L=\emptyset$ and $F_R=\emptyset$.)
\item
If there is no configuration in $F$ with the first component $e(J_1^i(r))$, then insert in $F$ the  dominating configuration in the pair $\{ (e(J_1^i(r)),\bar{x}_2;\bar{z}+1),\linebreak (e(J_1^i(r)),\hat{x}_2;\hat{z}) \}$.
\QED
\end{enumerate}
\end{procedure}

The analogous case where the new point $p_i$ is a non-buddy point such that $J^i_1(r)= \emptyset$ and $J^i_2(r)\neq \emptyset$ can be implemented similarly. Hence we are left with the case where $p_i$ is a buddy point, i.e.,  $J^i_1(r)\neq \emptyset$ and $J^i_2(r)\neq \emptyset$. This case can be handled by applying the approach from Procedure {\tt NonBuddy} on both $L_1$ and $L_2$ simultaneously. 
In particular, we scan $F$ twice, 
from left to right and right to left on $L_1$ and then from left to right and right to left on $L_2$, 
until we reach $l(J_1^i(r))$, $e(J_1^i(r))$ and $l(J_2^i(r))$, $e(J_2^i(r))$, respectively. Then we delete all the configurations that were scanned twice. Furthermore, we consider two pairs of new configurations as in Step~3 of Procedure~\ref{proc:nonBuddyL1}, one for each line $L_1$ and $L_2$. Among these at most four configurations, we add to $F$ those that are not dominated.
See Fig.~\ref{fig:nonBuddyBuddy}(b),
where non-bold lines must be deleted, and dashed lines are considered to be added to $F$.
\hide{
\begin{figure}[h]
\vspace{-23pt}
\centering
\includegraphics[height=27mm]{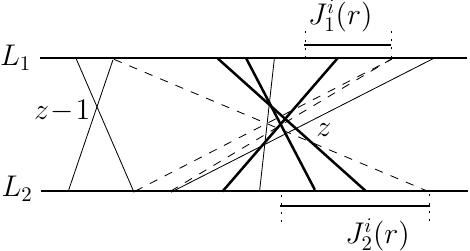}
\vspace{-8pt}
\caption{Updating the frontier configurations in the case of buddy intervals. Non-bold lines must be deleted, and dashed lines are considered to be added to $F$.}
\label{fig:Buddy}
\vspace{-10pt}
\end{figure}
}
\begin{lemma}\label{lem:feasibility1}
Algorithm~\ref{alg:feasibility} with the procedures defined above, can be implemented so that it can test $(k,r)$-feasibility in $O(n\log n)$ time.
\end{lemma}
\begin{proof}
Our procedures create at most two new configurations for each new point $p_i$ introduced, hence the total number of configurations created is $O(n)$.

We can keep two lists of pointers to configurations of $F$, one which is sorted with respect to the first component, and second one with respect to the second component. We implement these lists using a data structure that allows search, insert and delete in $O(\log n)$ time, for example using a 2-3 tree. Then it is easy to see that the scanning part of {\tt NonBuddy} procedures (Step~1) can be done in $O(m \log |F|)$ time, where $m$ is the number of removed (scanned) configurations. Steps 2 and 3 can be done $O(\log |F|)$ time.
Now let us consider the more complicated non-buddy case. 
As can be seen in Fig.~\ref{fig:nonBuddyBuddy}(b), there could be multiple subsequences of configurations that need to be removed, some of which are surrounded by configurations that need to be preserved. To reach areas that need to be removed we need $O(\log |F|)$ time, and to delete configurations we need $O(m \log |F|)$ time, where $m$ is the number of configurations that need to be removed. Hence, in both buddy and non-buddy cases, the processing of point $p_i$ can be done in $O(m_i \log n)$ time, where $m_i$ is the number of configurations deleted in the processing of $p_i$. Once deleted, a configuration will not be considered again, hence $\sum_{i=1}^n m_i=O(n)$. Therefore the complexity of Algorithm~\ref{alg:feasibility} is $O(n\log n)$.
\QED
\end{proof}

\noindent
As in~\cite{bereg2015b},
we can use Megiddo's method~\cite{megiddo1983a}
applied to the AKS sorting network~\cite{ajtai1983} with Cole's speed-up~\cite{cole1987}
 to prove
\begin{theorem}\label{thm:P2k}
When the centers are constrained to two parallel lines,
we can solve the $k$-center problem  in $O(n\log^2 n)$ time.
\QED
\end{theorem}

\section{Centers placed on both axes}\label{sec:perpendicular}
In this section we discuss only the unweighted case.\footnote{If the lines are perpendicular
to each other, 
the weighted case is more difficult to solve than the unweighted case.
For example,
unlike in the unweighted case,
four centers may not be sufficient to cover all the points in the ``center square,''
if there is a heavy point near the origin (0,0).
}
We assume the line $L_1$ (resp. $L_2$) to be the $x$- (resp. $y$)-axis,
and draw four lines $y=\pm r$ and $x=\pm r$,
as shown in Fig.~\ref{fig:perpendicular}.
\begin{figure}[ht]
\vspace{-12pt}
\centering
\subfigure[]{\includegraphics[height=41mm]{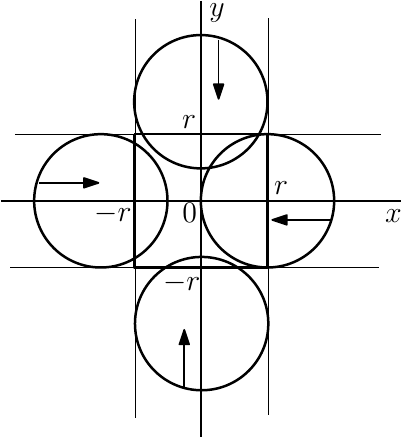}}
\hspace{1cm}
\subfigure[]{\includegraphics[height=37mm]{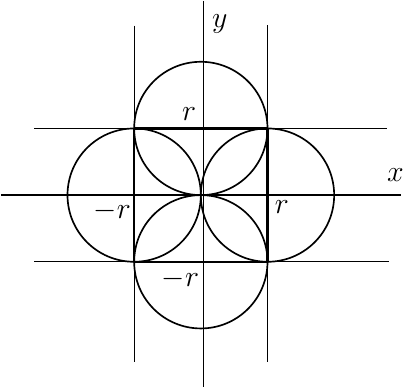}}
\vspace{-7pt}
\caption{
(a) Cover the points outside the center square $S$;
(b) Four circles can cover all the points in the $2r\times 2r$ center square.}
\label{fig:perpendicular}
\vspace{-7pt}
\end{figure}
To be covered by circles of radius $r$ on the two axes,
the given points in $P=\{p_1, p_2,\ldots, p_n\}$ must be contained
in the horizontal and/or vertical bands of width $2r$ defined by those four lines.
We first sort them according to both the $x$- and $y$-coordinates.
Clearly,
any point $p_i$ with $p_i.x >r$ or $p_i.x <-r$ (resp. $p_i.y >r$ or $p_i.y <-r$)
must be covered by a circle centered on the $x$-axis (resp. $y$-axis).
Using a greedy method,
we can find the minimum number of enclosing circles from outside-in on the $x$-axis and $y$-axis.
See Fig.~\ref{fig:perpendicular}(a).
So, assume that we have covered all those points,
introducing fewer than $k$ circles.
Thus we need to consider only the remaining points that lie within the square $S$
surrounded by four lines $y=\pm r$ and $x=\pm r$.
Note that the circles used to cover points outside $S$ may also cover
some points in $S$,
 and they should be excluded.

It is clear that all the points in $S$ can be covered if we use four circles,
as shown in Fig.~\ref{fig:perpendicular}(b).
We thus want to test if they can be covered by one, two or three circles of radius $r$.
By a greedy method,
we can easily test if all the $n$ points in $S$ can be covered by one, two or three circles
on the same axis in $O(n)$ time.
Therefore, without loss of generality,
we assume that one center is placed on the $y$-axis, and the others (one or two)
are placed on the $x$-axis.
We do some preprocessing to reduce the time needed in each ``round.''

\subsection{Interval tree}
For a given radius $r$,
the $2n$ end points of the $n$ given intervals on the $x$-axis
have a certain order from left to right.
If that order is known,
we can, in $O(n)$ time, construct a balanced {\em interval tree} $T_r$ with $2n$ leaves,
each representing an end point.
We can imagine that $T_r$ is built on top of the $x$-axis,
starting from the left end point of the leftmost interval to the right end point of
the rightmost interval.
Fig.~\ref{fig:intervaltree} illustrates a simple example,
in which there are four intervals on the $x$ axis,
named $a,b,c$, and $d$,
reflecting four points in set $P$.
\begin{figure}[ht]
\vspace{-7pt}
\centering
\includegraphics[height=26mm]{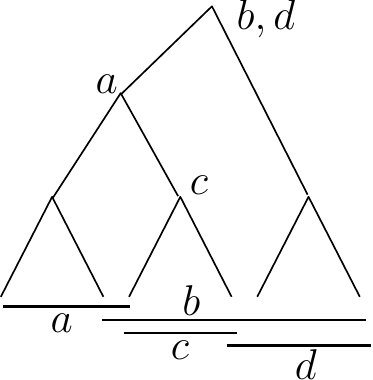}
\vspace{-7pt}
\caption{Interval tree $T_r$.}
\label{fig:intervaltree}
\vspace{-12pt}
\end{figure}

We store a set of 0 or more interval names at each node of $T_r$ as follows.
An interval is {\em active}, 
if it is not pierced by the center that we place on the $y$-axis.
Intervals may become active or inactive as the center on the $y$-axis is pushed down.
The set of points covered by the center undergoes changes $O(n)$ times.
For each interval $I$ that is initially active or becomes active,
we traverse the two upward path towards the root $\rho(T_r)$ of $T_r$ from its left (resp. right)
end point and paint the path blue (resp. red).
We store $I$ at the node where the blue and red paths meet.
Thus each interval appears at a unique node.
An edge may be painted blue, red, both, or not painted.
How many times an edge is painted by the same color is not important.

We now want to find the rightmost piercing point $q_l$ 
that pierces all the intervals whose left end points lie to the left of $q_l$.
It is easy to see that it can be found at the end (at a leaf node) of the leftmost red path from $\rho(T_r)$.
The procedure below shows how to update the coloring of $T_r$,
as intervals become active or inactive.

\begin{procedure} \label{proc:update}
Let $I$ be the interval that becomes activate/inactivate.
\begin{enumerate}
\item
~{\em [On becoming active]}
Paint the path from the left (resp. right) end point of $I$ to the root $\rho(T_r)$ blue (resp. red).
Let $v$ be the node where these two paths meet.
\item
~{\em [On becoming inactive]}
Remove the color due to the path from left (resp. right) end point of $I$ to the root $\rho_T$ until
the same-colored path due to some other interval is encountered.
\QED
\end{enumerate}
\end{procedure}

Let $q_l$ (resp. $q_r$) be the first piercing point that is required,
 starting from the left (resp. right) end.
 
\begin{lemma}\label{lem:feasibility2}
Suppose that colored paths are maintained as above.
\begin{enumerate}
\item
The piercing point $q_l$  is at the leaf reached by moving down from $\rho_T$ along the leftmost red path.
\item
The piercing point $q_r$  is at the leaf reached by moving down from $\rho_T$ along the rightmost blue path.\QED
\end{enumerate}
\end{lemma}

Now that we have found $q_l$ and $q_r$,
we want to know if they pierce all the active intervals.
Clearly, it is the case if $q_l$ does not lie to the left of $q_r$.
In this case two centers are enough to cover all the points in $S$,
which are not covered by the center on the $y$-axis.
So assume that $q_l$ lies to the left of $q_r$.
In this case we want to know if there is an interval that starts and ends between $q_l$ and $q_r$.

\subsection{Shortest-interval tree}
To find if there is such an interval,
 we construct, in $O(n)$ time, another search tree, similar to $T_r$,
called the \em shortest-interval tree and is denoted by ${\cal T}_r$.
Let $\rho({\cal T}_r)$ be the root of ${\cal T}_r$.
Let $J^i(r)=[f_i(r),g_i(r)]$ be an interval on the $x$-axis,
which is the intersection of disk $d_i(r)$ and the $x$-axis.
We interpret $g_i(r)$ as the \em value of $f_i(r)$.
We first sort the elements in $\{f_i(r)\mid -r< f_i(r) <r\}$ and arrange them from left to right.
They comprise the leaves of ${\cal T}_r$,
which is a heap  based on the values of the leaves,
so that for any node $u$ of ${\cal T}_r$, the minimum value among the leaves of subtree ${\cal T}_r(u)$
is associated with $u$.
Given any two leaves, $f_i(r)$ and $f_j(r)$,
we precompute the heap ${\cal T}_r$ and use it to find if there is any interval $J^h(r)$ contained
in the interval $(f_i(r),f_j(r))$ on the $x$-axis in $O(\log n)$ time.

As the center placed on the $y$-axis is pushed down,
the set of active elements in $S$ changes.
Let us define a large value $M$ by
\[
M=\max\{f_i(r)\mid -r< f_i(r) <r\} +1.
\]
When an interval $J^i(r)$ becomes inactive,
we change its value to $M$, and update the values associated with the nodes on
path $\pi(f_i(r),\rho({\cal T}_r))$, which takes $O(\log n)$ time,
using standard heap operations.
When $J^i(r)$ becomes active again,
we restore its value to $g_i(r)$ and update the values along $\pi(f_i(r),\rho({\cal T}_r))$.
Note that an interval changes its active/inactive status at most twice.
\begin{lemma}\label{lem:centerSquare}
For a given $r$,
we can test $(k,r)$-feasibility for the points in the $2r\times 2r$ square $S$ centered at
the origin in $O(n\log n)$ time for $k\leq 4$.
\QED
\end{lemma}
\begin{theorem}\label{thm:piercing}
When the centers are constrained to be on two perpendicular lines,
we can test $(k,r)$-feasibility in $O(n\log n)$ time.
\QED
\end{theorem}

\subsection{Optimization}
We want to look for the smallest $r$, named $r^*$,
such that the given points are $(k,r)$-feasible.
Clearly,
if we sort all the end points (the leaves of $T_r$) then the left endpoints
(the leaves of ${\cal T}_r$) are automatically sorted.
We adopt Megiddo's approach \cite{megiddo1983a} sort all the end points in $O(\log^2 n)$ time, 
using the AKS sorting network~\cite{ajtai1983},
which can be reduced to $O(\log n)$ time with Cole's improvement \cite{cole1987}.
Thus Theorem~\ref{thm:piercing} implies
\begin{theorem}\label{thm:perpendicular}
When the centers are constrained to be on two perpendicular lines,
we can solve the $k$-center problem in $O(n\log^2 n)$ time.
\QED
\end{theorem}

\subsection{Corner and T-junction}\label{sec:corner}
Consider the problem in which the two constraining lines form a $90^{\circ}$ corner,
as shown in Fig.~\ref{fig:corner}(a).
Note that the points in area $A$ can be covered only by 
circles centered on the $y$-axis.
\begin{figure}[ht]
\vspace{-15pt}
\centering
\subfigure[]{\includegraphics[height=27mm]{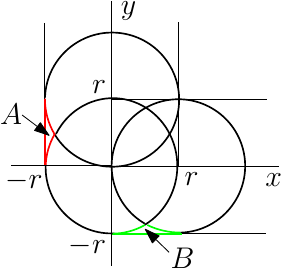}}
\hspace{2cm}
\subfigure[]{\includegraphics[height=31mm]{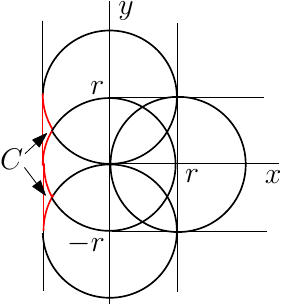}}
\vspace{-7pt}
\caption{(a) Corner; (b) T-junction.}
\label{fig:corner}
\vspace{-13pt}
\end{figure}
Therefore, in finding the minimum number of enclosing circles centered on the $y$-axis,
we must include them.
Similarly, the points in area $B$ can be covered only by 
circles centered on the $x$-axis,
thus in finding the minimum number of enclosing circles centered on the $x$-axis,
we must include them.
It is clear that the remaining points can be covered by at most three circles shown in
Fig.~\ref{fig:corner}(a).
Therefore, we can use the algorithm of Sec.~\ref{sec:perpendicular}.

If the two constraining lines form a T-junction,
as shown in Fig.~\ref{fig:corner}(b),
then the points in area $C$ can be covered only by 
circles centered on the $y$-axis.
Again, we can use the algorithm of Sec.~\ref{sec:perpendicular} in this case.
\begin{theorem}\label{thm:corner}
When the constraining lines form a $90^{\circ}$ corner or a T-junction,
we can solve the $k$-center problem in $O(n\log^2 n)$ time. 
\QED
\end{theorem}

\section{Conclusion}\label{sec:conclusion}
We have considered the constrained $k$-center problem,
where the centers must lie on two lines,
which are either parallel or perpendicular.
When the points are weighted and the two lines are parallel to the $x$-axis,
we have presented a $O(n\log n)$ time piercing algorithm.
If the two lines are perpendicular,
even if they are unweighted,
our piercing test method also takes $O(n\log n)$ time. 
We feel that this could be carried out in $O(n)$ time,
but this remains a conjecture at this time.
Another open problem is the weighted case for two perpendicular lines. 
The case where the centers are constrained to more than two lines is also open.
It is known that AKS sorting network is too huge to be practical,
but the recent result by Goodrich~\cite{goodrich2014} gives us hope
that it may become practical in the not-too-distant future.
\section*{Acknowledgement}\label{sec:ack}
We would like to thank Hirotaka Ono of Kyushu University and Yota Otachi of JAIST
for stimulating discussions on the topic of Section~\ref{sec:perpendicular}.
This work was supported in part by Discovery Grant \#13883 from
the Natural Science and Engineering Research Council (NSERC) of Canada and in part by MITACS,
both awarded to Bhattacharya.


\end{document}